\documentclass[a4paper,UKenglish,cleveref, autoref,thm-restate]{eptcs}
\usepackage{mathtools,mathpartir,amsmath,amssymb}
\usepackage[dvipsnames]{xcolor}
\usepackage{hyperref}
\usepackage{ccaption}
\usepackage{keyval}
\usepackage{prooftree}
\usepackage{soul}
\usepackage[all]{xy}
\usepackage{dsfont}
\usepackage{tikz}                   % Graphs
\usetikzlibrary{shadows,arrows,shapes,automata,positioning,decorations.markings,fit}
\usepgflibrary{arrows.spaced}
\usepackage{wrapfig}
\usepackage{breakurl}
\usepackage{stmaryrd}
\usepackage[misc,geometry]{ifsym}
\SetSymbolFont{stmry}{bold}{U}{stmry}{m}{n}
\usepackage{enumerate} 
\usepackage[T1]{fontenc}
\usepackage{orcidlink}

\title{Sensible Intersection Type Theories 
} 

\author{
Mariangiola Dezani-Ciancaglini 
\orcidlink{0000-0002-3341-0941}
\institute{
Dipartimento di Informatica,
Universit\`a di Torino, Italy}
\email{dezani@di.unito.it}
\and Besik Dundua 
\orcidlink{0000-0003-4754-4163}
\institute{Kutaisi International University, Kutaisi, Georgia\\
VIAM, Tbilisi State University, Tbilisi, Georgia}
\email{ Besik.Dundua@kiu.edu.ge}
\and
Paola Giannini\footnote{This work has the financial support of the Universit\`a  del Piemonte Orientale.}
\orcidlink{0000-0003-2239-9529}
\institute{DiSSTE, Universit\`{a} del Piemonte Orientale,  Italy}
\email{paola.giannini@uniupo.it}
\and
Furio Honsell
\orcidlink{0000-0001-8937-1892}
\institute{DSMIF, Universit\`a di Udine, Italy}
\email{furio.honsell@uniud.it}
}

\def\authorrunning{Dezani-Ciancaglini, Dundua, Giannini, Honsell}
\def\titlerunning{Sensible Intersection Type Theories}

\def \bpc {\begin{color}{blue}Paola:\ } 
\def \epc {\end{color}}

\newcommand{\CA}[1]{{\mathcal C}(#1)}
\newcommand{\cl}[2]{\gamma(#1,#2)}
\newcommand{\cls}[2]{[#1]^#2}
\newcommand{\leqc}{\preceq_{\Eqt}}

\newcommand{\B}{{\mathcal B}}

\newcommand{\gitt}{\Theta}
\newcommand{\gleq}{\sqsubseteq}
\newcommand{\gcap}{\sqcap}
\newcommand{\gto}{\rightsquigarrow}
\newcommand{\agi}{\alpha}
\newcommand{\bgi}{\beta}
\newcommand{\gder}[3]{#1\vdash_\gitt#2:#3}

\newcommand{\gb}{\Upsilon}
\newcommand{\gm}{\kappa}
\newcommand{\sder}[3]{#1\vdash_{\SL}#2:#3}
\newcommand{\SL}{\mathcal S\!\!\Lambda}

\usepackage{calrsfs}
\DeclareMathAlphabet{\pazocal}{OMS}{zplm}{m}{n}

\newtheorem{definition}{Definition}[section]
\newtheorem{lemma}[definition]{Lemma}
\newtheorem{proposition}[definition]{Proposition}
\newtheorem{theorem}[definition]{Theorem}

\newtheorem{example}[definition]{Example}

\newenvironment{proof}{{\em Proof.}}{\hfill$\square$\vspace{5pt}}
\newenvironment{proofs}{{\em Proof.}}{}

\newcommand{\refToEx}[1]{Example~\ref{#1}}

\newcommand{\Pair}[2]{\langle#1,#2\rangle}
\newcommand{\Triple}[3]{\langle#1,#2,#3\rangle}
\newcommand{\set}[1]{\{#1\}}

\newcommand{\Cline}[1] {$$#1$$} %{~\\[0.5mm]\centerline{$ #1 $}~\\[0.5mm]} %{\centerline{$#1$}\vspace{2pt}}

\newcommand{\la}{\lambda}     
\newcommand{\La}{\mathrm\Lambda}

 \newcommand{\U}{{\sf U}}
 \newcommand{\type}{A}%{\tau}
 \newcommand{\types}{B}%{\sigma}
 \newcommand{\typeC}{C}
  \newcommand{\typeD}{D}
   
\newcommand{\vartype}{{\sf c}}%{\phi}

\newcommand{\cSet}{\mathbb{A}}
\newcommand{\IT}{\mathbb{T}}
\newcommand{\ITT}{\pazocal{T}}
\newcommand{\der}[3]{#1\vdash_\ITT#2:#3}

\newcommand{\deri}[4]{#1\vdash_{#2}#3:#4}
 \newcommand{\rn}[1]{\textnormal{#1}}  
\newcommand{\M}{M}
\newcommand{\N}{N}
 \newcommand{\lm}[3]{\langle#1,#2,#3\rangle}
 \newcommand{\x}{x}
 \newcommand{\y}{y}
\newcommand{\dd}{{\sf d}}
 \newcommand{\subs}[3]{#1[#2:=#3]}
  
\newcommand{\D}{{\mathcal D}}
\newcommand{\flt}{\mathit{F}}  % a filter
\newcommand{\fltg}{\mathit{G}}  % a filter
\newcommand{\Flt}{\pazocal{F}} % filter model
 \newcommand{\env}{\rho} 
  \newcommand{\Env}[1]{\mathbb{E}_{#1}} 
 \newcommand{\isWFEnv}[2]{#1\models#2} 
   
\newcommand{\semD}[3]{\llbracket#1\rrbracket^{#2}_{#3}}
 \newcommand{\appF}{\cdot} 
 
\newcommand{\bred}{\to_\beta}
\newcommand{\bredstar}{\to^\ast_\beta}
\newcommand{\hred}{\to_h}
\newcommand{\hredstar}{\to^\ast_h}

\newcommand{\ie}{{\em i.e.}}

\newcommand{\So}{\mathcal S}
\newcommand{\Bo}{\mathcal B}
\renewcommand{\Bo}{\pazocal B}
\newcommand{\aenv}{\zeta_\cSet}

\newcommand{\tint}[2]{[#1]_{#2}}
\newcommand{\tseq}[2]{\langle#1\mid #2\rangle}
\newcommand{\tseqs}[1]{\langle#1\rangle}
\newcommand{\op}[1]{{\mathcal O}_{#1}}
\newcommand{\vs}[1]{\mathbb{#1}}
\newcommand{\SAT}{\pazocal{S\!A\!T}}
\newcommand{\Opp}[1]{#1^{\it op}}

\newcommand{\pI}{I}%{\pazocal{I}}
\newcommand{\pJ}{J}%{\pazocal{J}}
\newcommand{\leqS}{~^\subseteq_\supseteq}
\newcommand{\transl}[1]{#1^\star}%{[\![#1]\!]}

\newcommand{\Sat}{\pazocal{X}}

\newcommand{\pol}{p}
\newcommand{\piu}{{\tt +}}
\newcommand{\meno}{{\tt -}}
\newcommand{\pom}{{\tt \pm}}
\newcommand{\Pol}[2]{{#1}^{#2}}
\newcommand{\PN}[1]{\Pol#1{\pom}}
\newcommand{\Pos}[1]{\Pol#1{\piu}}
\newcommand{\Neg}[1]{\Pol#1{\meno}}

\newcommand{\OkPos}{\mathbf{Pos}}
\newcommand{\OkNeg}{\mathbf{Neg}}

\newcommand{\Eqt}{\mathcal{A}}

\newcommand{\fp}[1]{\mathbf{#1}}

\begin{document}

\maketitle

\begin{abstract} Finitary/static semantics in
 the form of intersection type assignments have become a paradigm for analysing the fine structure of all sorts of $\lambda$-models. The key step is the construction of a filter model isomorphic to a given $\lambda$-model. A property of great interest of filter $\lambda$-models is {\em sensibility}, \ie\ the interpretation of all unsolvable terms is the least element.  
The flexibility of intersection type assignments derives from their parametrisation on intersection type theories. 
We construe intersection type theories as special meet-semilattices and show that appropriate morphisms, in the opposite category of meet-semilattices, preserve sensibility of the induced $\lambda$-models. Interestingly the set of saturated sets together with the set of $\lambda$-terms is such a meet-semilattice, thus showing that arguments based on Tait-Girards's computability amount to the construction of a morphism. We characterise two classes of intersection type theories which induce sensible filter models. The first is non-effective while the second is effective and it amounts to the generalisation of Mendler's criterion to intersection types and head normalising terms. The complete characterisation of sensible filter models however still escapes. 
\end{abstract}
{\bf Keywords}: $\la$-calculus, Intersection Types, Filter Models.
\bigskip
{\begin{flushright} \small a Stefano Berardi {\em il miglior fabbro} \end{flushright}}

\section{Introduction}
% !TEX root = dgh.tex

Dedicating a paper to a distinguished colleague is already quite a demanding task, but this task becomes even harder if our colleague has spread his remarkable talent across the two fields of Logic and Theoretical Computer Science. The  inspiration for the topic of the present paper came from the recollection of Stefano Berardi as a young PhD student in Torino when, in the Stone Age  of Logical Frameworks and the early days of the Types Community, he addressed the problem of formal machine checking the proof given by J.-Y. Girard's of the strong normalisation of second order $\lambda$-calculus~\cite{G71}, using  one of the first releases of Coq~\cite{CH88}. Stefano went well beyond that and since then made momentous contributions to the area of extracting constructive contents from classical and impredicative proofs~\cite{BB95,B06,BS17,BT19,BBR25}.

The present paper addresses a closely related problem, namely that of {\em head %weak 
normalisation} for {\em intersection type theories}.  Intersection type theories~\cite{CD80} were invented in Torino by the first author, together with Mario Coppo, in the late '70's of the last century of the previous millennium. As already noticed in~\cite{CDHL84} intersection types are information systems in the sense of~\cite{S82}. Since then, intersection types have been widely generalised and utilised for providing useful  characterisations for several classes of $\lambda$-terms, most notably weak head normalising~\cite{DHM05}, head normalising~\cite{CDV81}, normalising~\cite{CDV81} and their persistent versions~\cite{DHM05}, strongly normalising~\cite{P80}, closable~\cite{HR92}, and invertible $\lambda$-terms~\cite{T08,V19}. The flexibility of intersection types lies in their correspondence with {\em clopen sets} in Scott's topological models of $\lambda$-calculus, which can thus be understood as models whose points are, in fact, filters of properties of programs. Intersection type theories therefore permit to express the dynamics of programs as filters of their static properties~\cite{BCD83,CDHL84}. This correspondence has been nicely expressed categorically as a duality in~\cite{A91}.  Since their introduction, intersection types have become a paradigm for expressing statically all sorts of  execution properties of programming languages~\cite[Part III]{BDS13}. 

Intersection type theories being so flexible, which in fact is the very reason which makes them successful, are far from having a complete theory. More specifically, in the present paper we address the problem of characterising {\em sensible} intersection type theories, namely type theories which generate {\em sensible} filter models, \ie\ models which assign only the trivial intersection type to an unsolvable term. To this end we construe intersection type theories as meet-semilattices, enriched with an arrow constructor,  and show that appropriate morphisms in the opposite category of meet-semilattices preserve sensibility of the induced $\lambda$-models. This  permits us to  transfer profusely sensibility results between filter models, thus providing alternatives to the existing  proofs of sensibility for many models~\cite{DHM05}. 
The  very set of saturated sets, together with the set of $\lambda$-terms, being such a meet-semilattice, permits us to reduce to the existence of a morphism all arguments based on Tait-Girards's computability, as was the one formalised by Berardi some thirty-five years ago now. This is in effect a generalisation of Girard's reducibility candidates.  We characterise two classes of sensible intersection type theories. The first is non-effective and it applies to 
%the special class of $\rightarrow$-sound intersection type theories. 
a special class of intersection type theories satisfying a technical condition known as $\to$-soundness.
The second is effective and it amounts to the generalisation of Mendler's criterion~\cite{M91}, originally given for recursive second order $\lambda$-calculus and strong normalisation, to intersection types and head normalisation. The complete characterisation of sensible filter models however still escapes. 

The present paper is a follow up of~\cite{DGH25}, where the complementary problem of studying non-sensible intersection type theories was addressed. Reading both papers can be beneficial, since the two papers have a number, albeit small, of cross-references.

Finally, we wish Stefano Berardi, the {\em miglior fabbro}\footnote{In the XXVI canto of {\em Purgatory} by Dante Alighieri, Guido Guinizzelli indicates Arnaut Daniel (Occitan troubadour of the 12th century) as ``the best smith of maternal speech'' (``il miglior fabbro del parlar materno'') for his poetic mastery. T. S. Eliot dedicated the final version of {\em The Waste Land} to Ezra Pound, calling him ``il miglior fabbro''.}, an even more fruitful late career in $\lambda$-calculus in the tradition of such luminaries as Curry, Church, B\"ohm, Scott, Martin-L\"of, Venturini-Zilli, Ronchi Della Rocca, Plotkin, and Barendregt, $\ldots$.  But  we also hope that the present paper can stimulate other authors to take up the fascinating task of clarifying further the mysteries of intersection types. Since, ultimately, given that $\lambda$-calculus is a universal model of computation,  these are the mysteries of computation itself.

\subsubsection*{Synopsis} In Sections~\ref{lc} and~\ref{itfm}, we recall basic facts on $\lambda$-calculus and the theory of intersection types and filter models. In Section~\ref{gitts}, we introduce the algebraic framework of meet-semilattices, establish transfer results for sensibility between theories, and provide illustrative examples.
%we introduce the algebraic setting of meet-semilattices and establish the transfer results of sensibility between theories and give examples. 
In Section~\ref{gt}, we  characterise  two classes of intersection type theories which are sensible. Difficulties in providing complete characterisation of sensible intersection type theories appear in Section~\ref{critical}, where we also discuss the $\lambda$-theories of sensible filter models and raise some open questions. Concluding remarks appear in Section~\ref{rwc}.

\section{$\lambda$-calculus}\label{lc}
% !TEX root =dgh.tex

In this section we recall some basic notions and properties of untyped $\la$-calculus following Chapters 2, 3, and 8 of~\cite{B85}. Readers familiar with $\lambda$-calculus 
%these chapters 
can skip this subsection.

We start by defining $\la$-terms and $\beta$-reduction.

\begin{definition}[$\la$-terms~{\cite[Definition 2.1.1]{B85}}]\label{lt}
The set $\La$ of  {\em pure $\la$-terms} is defined by:
 \Cline{M::= x\mid\la  x.  M\mid  M M.}
\end{definition}

\noindent
We write $\lambda$-terms with the usual notational conventions.  In particular we write $\lambda\overrightarrow{x}. M$ as short for 
$\lambda x_1\cdots  x_n. M$ assuming $\overrightarrow{x}= x_1\cdots  x_n$ for $n\in {\mathbb N}$.  Free and bound occurrences of variables are defined in the standard way. In particular we assume Barendregt's convention, \ie\  that different variables have different names~{\cite[Convention 2.1.12]{B85}}.

\begin{definition}[$\beta$-rule and $\beta$-reduction~{\cite[Definitions 2.1.15, 3.1.3 and 3.1.5]{B85}}]\hspace{-3pt}
\begin{enumerate}
\item The {\em $\beta$-rule} replaces 
$(\la x. M)N$ with $ M[x:=N]$, 
where $M[x:=N]$ denotes the $\la$-term obtained by the (capture free) substitution of $x$ by $N$ in $M$.
\item The {\em one step $\beta$-reduction} $\bred$ is defined as the contextual closure of the $\beta$-rule.
\item The {\em $\beta$-reduction} $\bredstar$ is defined as the reflexive and transitive closure of $\bred$. 
\item The {\em $\beta$-convertibility} $=_\beta$ is defined as the equivalence relation generated by $\bredstar$. 
\end{enumerate}
\end{definition}

Crucial to our development are the notions of solvability and unsolvability of $\la$-terms.

\begin{definition}[Solvable and unsolvable $\la$-terms~{\cite[Definition 2.2.10]{B85}}]\hspace{-3pt}
\begin{enumerate}
\item A $\la$-term $M$ is {\em solvable} if there are $n$ $\la$-terms $N_1, \ldots, N_n$ such that
\Cline{(\la\overrightarrow x.M)N_1 \cdots N_n \bredstar  {\bf I},}
where $\overrightarrow x$ are the variables which occur free in $M$ and ${\bf I}=\la x.x$ is the identity combinator.
\item A $\la$-term is {\em unsolvable} if it is not solvable.
\end{enumerate}
\end{definition}

As in~\cite{K98} our study of unsolvable terms is based on the notion of head reduction.

\begin{definition}[Head normal form and head redex~{\cite[Definition 8.3.9]{B85}}]\hspace{-3pt}
\begin{enumerate}
\item  If $M=\la\overrightarrow x.xM_1\cdots M_m$, then $M$ is in {\em head normal form} and 
$x$ is the {\em head variable} of $M$.
\item If $M=\la\overrightarrow x.(\lambda x.N)PM_1\cdots M_m$, then $(\lambda x.N)P$ is the {\em head redex} of $M$.
\end{enumerate}
\end{definition}

Every $\la$-term either is in head normal form or has a head redex.

\begin{proposition}[Shape of $\la$-terms~{\cite[Corollary 8.3.8]{B85}}]
Every $\la$-term is   

\centerline{$\text{either of the form }\la \overrightarrow x.x M_1\cdots M_m \text{ or of the form } 
\la \overrightarrow x.(\lambda x.N)P M_1\cdots M_m$} 

\noindent
where $m\geq 0$.
\end{proposition}

\begin{definition}[Head reduction~{\cite[Definition 8.3.10]{B85}}] We write $M \hred N$ if $N$ is obtained from $M$ by reducing its head redex.
The {\em head reduction} of $M$ is the finite or infinite sequence of terms $M_0$, $\dots$, $M_n$, $\dots$ such that $M=M_0$ and $M_{n} \hred M_{n+1}$  with $n\in \mathbb{N}$. 
\end{definition}

\noindent
We use $\hredstar$ to denote the reflexive and transitive closure of $\hred$.

\medskip

In our development we take advantage of the characterisation of unsolvability by means of head reduction. 
\begin{theorem}[{\cite[Fact 2.2.12]{B85}}]
A $\la$-term $M$ is unsolvable iff its head reduction is infinite.
\end{theorem}

\section{Intersection Types and Filter Models}\label{itfm}
% !TEX root =dgh.tex

This section is devoted to the definitions of intersection types, type theories, type assignment systems and filter models. 

Up to Definition~\ref{tas} (included) we essentially follow Sections 13.1 and 13.2 of~\cite{BDS13}. The only differences are that, in defining intersection types and subtyping,   we require
 the constant $\U$, which is optional in~\cite{BDS13},  and our subtyping relation has the additional Axiom (\rn{$\U_{top}$}) and Rule \rn{($\to^\sim$)}. 

\begin{definition}[Intersection Type Theories]\label{itt's}\hspace{-3pt}
\begin{enumerate}
\item Given a set of constants $\cSet$ and a distinguished constant $\U$, the set $\IT_\cSet$ of {\em intersection types} over $\cSet\cup\set{\U}$ is generated by the grammar:
\Cline{\type::=\vartype\mid\U\mid\type\to\type\mid\type\cap\type,}
where $\vartype\in\cSet$. 
\item A {\em subtyping relation}  $\leq$ is a binary relation  on $\IT_\cSet$ closed under the following axioms and rules: 
\Cline{\begin{array}{ccccccc}\type\leq\type\ \rn{(Refl)}&\qquad&\types\cap\type\leq\types\ \rn{(IncL)}&\qquad&\types\cap\type\leq\type\ \rn{(IncR)}&\qquad&\type\leq\U\ \rn{($\U_{top}$)}
\end{array}}
\Cline{\begin{array}{ccccc}
\prooftree
\types\leq\type\quad\types\leq\type'
\justifies
\types\leq\type\cap\type'
\using \rn{(Glb)}
\endprooftree&\qquad&
\prooftree
\types\leq\type\quad\type\leq\type'
\justifies
\types\leq\type'
\using \rn{(Trans)}
\endprooftree&\qquad&
\prooftree
\types'\sim\types\quad\type\sim\type'
\justifies
\types\to\type\sim\types'\to\type' 
\using \rn{($\to^\sim$)}
\endprooftree\end{array}}
where $\type\sim\types$ is short for  $\type\leq\types$ and $\types\leq\type$.
\item An {\em intersection type theory (itt)} $\ITT$ is determined by a set of type constants $\cSet$ and a subtyping relation on the set $\IT_\cSet$, \ie\  
$\ITT=\Pair{\cSet}{\leq_\ITT}$.
\end{enumerate}
\end{definition}

\noindent
We  adopt the convention  that $\cap$ has precedence over $\to$.  
The above rules imply that $\cap$ preserves the congruence of its arguments 
w.r.t. $\sim_\ITT$ and moreover that it %$\cap$ 
is  idempotent, commutative and associative with neutral element $\U$. Moreover Rule  \rn{($\to^\sim$)} implies that $\to$ preserves the congruence of its arguments w.r.t. $\sim_\ITT$. This rule is less demanding than the usual covariance/contravariance  of the arrow given in Rule  \rn{($\to$)} of Figure~\ref{fig:axRules}. 
We assume that $\bigcap_{i\in\emptyset}\type_i=\U$.  We summarise this with a proposition which will be useful in Section~\ref{gitts}.

\begin{proposition}\label{semilattice} The equivalence classes of an itt $\ITT$ w.r.t. the equivalence $\sim_\ITT$ define a 
meet-semilattice enriched with a binary arrow constructor. 
\end{proposition}
%We assume $\cap$ to be idempotent, commutative and associative with neutral element $\U$ and $\bigcap_{i\in\emptyset}\type_i=\U$. 

\begin{definition}[Type Assignment System]\label{tas}
The {\em intersection type assignment system} induced by an itt $\ITT=\Pair{\cSet}{\leq_\ITT}$ is a formal system deriving judgements of the shape
$\der\Gamma M\type$, where $\type\in\IT_\cSet$ and a {\em basis} $\Gamma$ is a finite mapping from term variables to types in $\IT_\cSet$:
\Cline{\Gamma::=\emptyset\mid\Gamma, x:\type.} 
\noindent
The axioms and rules of the type system are the following, where by writing $\Gamma, x:\type$ we assume that $x$ does not occur in $\Gamma$.

\Cline{\begin{array}{ccc}
\prooftree
\justifies
\Gamma,x:\type\vdash x:\type\using\rn{(Ax)}
\endprooftree
&\qquad&
\prooftree
\justifies\Gamma\vdash M:\U\using\rn{($\U$)}
\endprooftree
\\[5mm]
\prooftree
\Gamma, x:\types\vdash M:\type
\justifies
\Gamma\vdash\lambda  x. M:\types\to\type
\using\rn{($\to$I)}
\endprooftree
&&
\prooftree
\Gamma\vdash M:\types\to\type\quad\Gamma\vdash N:\types
\justifies
\Gamma\vdash M N:\type
\using\rn{($\to$E)}
\endprooftree\\[7mm]
\prooftree
\Gamma\vdash M:\types\quad\Gamma\vdash M:\type
\justifies
\Gamma\vdash M:\types\cap\type
\using\rn{($\cap$I)}
\endprooftree
&&
\prooftree
\Gamma\vdash M:\types\quad \types\leq_\ITT\type
\justifies
\Gamma\vdash M:\type
\using\rn{($\leq$)}
\endprooftree
\end{array}}
\end{definition}

It is easy to verify that the following rules are admissible
\Cline{\prooftree
{\Gamma,x:\types}\vdash M: \type \quad\typeC\leq_\ITT\types
\justifies
{\Gamma,x:\typeC}\vdash M :\type
\using \rn{($\leq$-L)}
\endprooftree\qquad\qquad
\prooftree
{\Gamma}\vdash M :\type \quad x\not\in \Gamma
\justifies
{\Gamma,x:\types}\vdash M :\type
\using \rn{(Weakening)}
\endprooftree}
where $x\not\in \Gamma$ is short for 
$x$ does not occur in $\Gamma$.

\bigskip

The main  properties of intersection type assignment systems are the Inversion Lemma and Subject Expansion,  which are proved by induction on type derivations. 

\begin{lemma}[Inversion Lemma~{\cite[Theorem 14.1.1]{BDS13}}]\label{il}\hspace{-3pt}
\begin{enumerate}
\item\label{il1} If $\der\Gamma x \type$ and $\type\nsim_\ITT\U$, then $\Gamma(x)\leq_\ITT \type$;
\item\label{il2} If $\der\Gamma {MN} \type$ and $\type\nsim_\ITT\U$, then there are $I$ and $B_i$, $C_i$ for $i\in I$ such that $\bigcap_{i\in I} C_i\leq_\ITT A$ and $\der\Gamma M {B_i\to C_i}$ and $\der\Gamma N {B_i}$ for all $i\in I$;
\item\label{il3} If $\der\Gamma {\la x.M} \type$, then there are $I$ and $B_i$, $C_i$ for $i\in I$ such that $\bigcap_{i\in I}(B_i\to C_i)\leq_\ITT A$ and \linebreak $\der{\Gamma,x:B_i} M {C_i}$  for all $i\in I$. 
\end{enumerate}
\end{lemma}

\begin{theorem}[Subject Expansion~{\cite[Corollary 14.2.5(ii)]{BDS13}}]\label{se}
$M\bred M'$ and $\der\Gamma {M'} \type$ imply \Cline{\der\Gamma {M} \type.}
\end{theorem}

Also crucial is the property of Subject Reduction, which however holds only with a proviso.

\begin{theorem}[Subject Reduction~{\cite[Proposition 14.2.1(ii)]{BDS13}}]\label{sr}
\Cline {M\bred M' \text{ and } \der\Gamma M \type \text{ imply } \der\Gamma {M'} \type} \begin{center} if and only if\end{center}
 \Cline{\der\Gamma{\lambda x.N}{\types\to\typeC} \text{ implies } \der{\Gamma,x:\types}{N}{\typeC}.}
\end{theorem}

\noindent 
In fact, not all   
 type systems  induced by itt's
 enjoy Subject Reduction.
Consider  $\ITT_0=\Pair{\set{\vartype_0,\vartype_1}}{\leq_{\ITT_0}}$, where $\ITT_0$ has only the axiom $\vartype_0\to\vartype_0\leq\vartype_1\to\vartype_0$, then  
 $\deri{}{\ITT_0}{\lambda x.x}{\vartype_1\to\vartype_0}$, but $x:\vartype_1\not\vdash_{\ITT_0}x:\vartype_0$. Subject Reduction fails since $\deri{y:\vartype_1}{\ITT_0}{(\lambda x.x)y}{\vartype_0}$, but $y:\vartype_1\not\vdash_{\ITT_0}y:\vartype_0$.

\medskip

A sufficient  but not necessary  condition for Subject Reduction is $\beta$-soundness. 

\begin{definition}[$\beta$-soundness~{\cite[Definition 14.1.4]{BDS13}}]\label{bs}
An  itt $\ITT$ is {\em $\beta$-sound} if $\type\nsim_\ITT\U$ and
$\bigcap_{i\in I}(\types_i\to\type_i)\leq_\ITT\types\to\type$ imply that there is $J\subseteq I$ such that $ \types\leq_\ITT\bigcap_{j\in J}\types_j$ and $\bigcap_{j\in J}\type_j\leq_\ITT\type$.
\end{definition}

\noindent
For example the itt $\ITT_1=\Pair{\set{\vartype_0,\vartype_1}}{\leq_{\ITT_1}}$, where $\leq_{\ITT_1}$ has no other axioms and rules, is 
$\beta$-sound. In contrast, the itt $\ITT_0$ defined above is not $\beta$-sound.  The itt $\ITT_0$ can be made $\beta$-sound by adding  the axiom  $\vartype_1\leq\vartype_0$.  
Two  itt's which are not $\beta$-sound but still satisfy Subject Reduction are defined in~\cite{CDHL84,ABD03}. 

\medskip

 We consider  two important classes of itt's which are defined and motivated in~\cite{DGH25} (Definitions 26 and 27) .

 \begin{definition}[Set Condition]\label{dg}
An itt $\ITT=\Pair{\cSet}{\leq_\ITT}$ satisfies the {\em set condition} if  \Cline{\text{$\bigcap_{i\in I}\type_i\leq_\ITT\types_1\to\cdots\to\types_n\to\typeC$ with $\typeC\nsim_\ITT\U$  implies $\type_j\sim_\ITT\types_1\to\cdots\to\types_n\to\typeD$}} for some $j\in I$ and some $\typeD\nsim_\ITT\U$ such that $\typeC\cap\typeD\sim_\ITT\typeC$. \end{definition}
\noindent
%We remark that the above condition cannot be taken as a ``rule of derivation'' in the system, but it is only a rule of proof, \ie\ an admissible rule.

\begin{figure}
\begin{center}
\begin{tabular}{ll}
 {\large $\U\sim\type\to\U$\quad\rn{($\to\U$)}} &
{\large $(\types\to\type)\cap(\types\to\type')\sim\types\to\type\cap\type'$ \quad\rn{($\to\cap$)}}
\\[9pt] 
{\large 
$\prooftree
\types'\leq\types\quad\type\leq\type'
\justifies
\types\to\type\leq\types'\to\type'\ \ 
\using \rn{($\to$)}
\endprooftree$}&\qquad
{\large $\prooftree
\U \leq B \to A
\justifies
\U\leq A
\using \rn{($\U\leq$)}
\endprooftree$} 
\end{tabular}
\end{center}
\caption{Some axioms and rules for itt's.}\label{fig:axRules}
\end{figure}

\begin{definition}\label{lgs}
Consider axioms and rules in Figure~\ref{fig:axRules}.  
\begin{enumerate}
\item\label{lgs2}  An itt  is {\em set-like} if  it satisfies the set condition and at least  Axioms \rn{($\to\U$)} and \rn{($\to\cap$)} hold. 
\item\label{lgs3} An itt is {\em $\to$-sound} if it satisfies at least  Axioms \rn{($\to\U$)}, \rn{($\to\cap$)}, and Rules \rn{($\to$)}, \rn{($\U\leq$)} hold.\end{enumerate}
\end{definition}

In order to discuss $\la$-models over itt's we recall the definition of $\la$-model.
An {\em environment} on the set $\D$ is a total mapping  from term variables to elements of $\D$. Let $\env$ range over environments. As usual, we denote by $\subs\env\x \dd$ the environment  which returns $\dd$ when applied to $\x$  and $\env(\y)$ when applied to $\y\not=\x$. 

\begin{definition}[$\la$-model~{\cite[Definition 16.1.2]{BDS13}}]\label{dlm} 
A {\em $\la$-model} is a triple $\lm\D\cdot{\semD{~}\D{}}$, where $\cdot$ is a binary operation on  $\D$ {\em (application)}, $\semD{~}\D{}$ is a mapping  from $\la$-terms and environments in $\D$ to elements of $\D$ {\em (term interpretation)},  and  
$\semD{~}\D{}$ satisfies:
\begin{enumerate}
\item $\semD\x\D\env=\env(\x)$;
\item $\semD{\M\N}\D\env=\semD\M\D\env \cdot\semD\N\D\env$;
\item $\semD{\la\x.\M}\D\env=\semD{\la\y.\M[\x:=\y]}\D\env$;
\item $\forall\dd\in\D. \semD\M\D{\subs\env\x \dd}=\semD\N\D{\subs\env\x \dd}$ implies $\semD{\la\x.\M}\D\env=\semD{\la\x.\N}\D\env$; 
\item $\env(\x)=\env'(\x)$ for all variables $\x$ which occur free in $M$ implies $\semD\M\D\env=\semD\M\D{\env'}$;
\item $\semD{\la\x.\M}\D\env\cdot\dd=\semD\M\D{\subs\env\x \dd}$.
\end{enumerate}
\end{definition}

\noindent
This definition of $\lambda$-model was first formulated by Hindley and Longo~\cite{HL80}.

\medskip

 We can build $\la$-models
whose domains  are sets of filters of types according to the following definition. 
%Filters of types are sets closed under intersection  and subtyping.

\begin{definition}[Filter~{\cite[Definition 13.4.1]{BDS13}}]
Let $\ITT=\Pair\cSet{\leq_\ITT}$ be an itt and
 $\flt\subseteq\IT_\cSet$. The set $\flt$ is a {\em $\ITT$-filter} if: 
\begin{itemize}
\item $\U\in\flt$;
\item $\type,\types\in\flt$ imply $\type\cap\types\in\flt$;
\item $\type\in\flt$ and $\type\leq_\ITT\types$ imply $\types\in\flt$.
\end{itemize}
\end{definition}

\noindent
We use $\flt$ and $\fltg$ as metavariables for filters and $\Flt_\ITT$ to denote the set of $\ITT$-filters. 
 If $X\subseteq \IT_\cSet$ we denote by $\uparrow_\ITT X$ the smallest $\ITT$-filter which contains $X$. If $X=\set \type$ we use $\uparrow_\ITT \type$ as short for $\uparrow_\ITT \set \type$.

Filters can be endowed with an applicative structure as follows:  
\begin{definition}[Filter Structure]\label{fs} Let $\Env{\ITT}$ be the set of environments on $\Flt_\ITT$. 
The {\em filter structure over $\ITT$} is  the triple $\Triple{\Flt_\ITT}{\appF}{\semD{~}{\Flt_\ITT}{}}$ where 
\begin{itemize}
\item {\em application}, $\appF:\Flt_\ITT\times\Flt_\ITT\to\Flt_\ITT$, is defined by
\Cline{  \flt \appF \fltg=\set{\type\mid \exists \types\in\fltg.\ \types\to\type\in\flt};}
\item {\em term interpretation}, $\semD{~}{\Flt_\ITT}{}:\La\times\Env{\ITT}\to\Flt_\ITT$, is defined by
\Cline{\semD{M}{\Flt_\ITT}{\env}=\set{\type\in\IT_\cSet\mid\exists\isWFEnv{\Gamma}{\env}.\ \der{\Gamma}{M}{\type}},}
where $\env$ ranges over $\Env{\ITT}$ and $\isWFEnv{\Gamma}{\env}$ if $x:\type\in\Gamma$ implies $\type\in\env(x)$.
\end{itemize}
\end{definition}

\begin{definition}[Filter Model] A {\em filter model} is a filter structure where all conditions of Definition~\ref{dlm} hold.
\end{definition}
It is easy to verify that $\semD{~}{\Flt_\ITT}{}$ satisfies all conditions required to be a $\la$-model (Definition~\ref{dlm}), but the last one, which is essential when $\dd$ is the interpretation of a $\la$-term.  We always have $\semD{\M[\x:=\N]}{\Flt_\ITT}\env\subseteq\semD{(\la\x.\M)\N}{\Flt_\ITT}\env$, since Subject Expansion holds by Theorem~\ref{se}. 

\begin{theorem}[{\cite[Proposition 16.2.4]{BDS13}}]\label{fsfm}
The filter structure over $\ITT$ is a filter model iff 
\Cline{\semD{(\la\x.\M)\N}{\Flt_\ITT}\env\subseteq\semD{\M[\x:=\N]}{\Flt_\ITT}\env} for all $\la$-terms $M,N\in\Lambda$, all variables $x$ and  all environments $\env$ in $\Flt_\ITT$.
\end{theorem}

The condition $\semD{(\la\x.\M)\N}{\Flt_\ITT}\env\subseteq\semD{\M[\x:=\N]}{\Flt_\ITT}\env$ means that  all types of $(\la\x.\M)\N$ are also types of  $\M[\x:=\N]$,  % have the same types, 
\ie\  that the type system $\vdash_\ITT$ enjoys Subject Reduction.  Then the following theorem follows naturally,  being $\beta$-soundness a sufficient condition for Subject Reduction.

\begin{theorem}[{\cite[Corollary 16.2.9(i)]{BDS13}}]\label{cfsfm}
If $\ITT$ is a $\beta$-sound itt, then the filter structure over $\ITT$ is a  filter model.
\end{theorem}

All set-like itt's generate filter models, since it is easy to check that the set condition implies $\beta$-soundness. 

As mentioned after Definition~\ref{bs},  in~\cite{CDHL84,ABD03} there are filter models  over  itt's which are not $\beta$-sound. 

It is interesting to notice that all continuous functions are representable in a filter model over a $\beta$-sound itt. This generalises Theorem 2.13(iii) in~\cite{CDHL84}.

Notably graph models~\cite{S76,P93} are isomorphic to filter models over set-like itt's and inverse limit models~\cite{S72,W76} are isomorphic to filter models over $\to$-sound itt's, see Example 35 in~\cite{DGH25}.

\bigskip

We conclude this section giving a crucial definition in this paper:
\begin{definition}[Sensible Itt, Sensible Filter Model]\label{def:sensible}
An itt $\ITT$ is {\em sensible} if all unsolvable terms are typed only by types equivalent to $\U$. A filter model $\Flt_\ITT$  is {\em sensible} if  $\ITT$ is sensible.  Otherwise the itt 
and the filter model are said to be {\em non-sensible}.
\end{definition}

Clearly in a sensible filter model all unsolvable terms are interpreted in the bottom filter  $\uparrow_\ITT \U$.

We remark that our notion of sensibility implies that the induced equational theory  is order-sensible as defined in~\cite[Definition 12.1(ii)(3)]{BM22}. Hence it is more restrictive than the standard condition on $\lambda$-models which only requires that all closed unsolvable terms are equated~\cite[Definition 4.1.7(ii)]{B85}. In Subsection~\ref{theories} we shall discuss the issue further.

\section{Transfer Theorems}\label{gitts}
                                         % !TEX root =dgh.tex

To the best of our knowledge the original proofs of head-normalisation for itt's, both historically and logically, are based on three methodologies: proof-normalisation~\cite{P65,P77,S78}, indexed reductions~\cite{L75}, or Tait-Girard reducibility arguments~\cite{T67,G71}. 
For the purpose of studying when itt's are sensible, once a given itt has been shown to be sensible, it is natural to try to design a setting in which this result can be easily transferred to similar itt's.
To this end it appears convenient to take a more abstract, and less language-dependent, view of itt's,  as in~\cite{DDH26}. We therefore  introduce below a notion of type structure, called {\em generalised intersection type theory} ({\em gitt}), together with a notion of morphism between such structures, which will allow for transferring  directly properties, such as sensibility, between type systems. We reckon this extension satisfactory, since the very proofs by Tait-Girard reducibility will appear as transfer results from the set of reducibility candidates viewed as generalised types, as will become apparent in Theorem~\ref{msts} and in the next section.

\begin{definition}[Generalised Intersection Type Theory]\label{gitt}
A {\em generalised intersection type theory} (shortly gitt) is a not trivial meet-semilattice $\Pair\gitt {\gleq_\gitt}$ with a top $\top_\gitt$ and closed under an arrow type constructor $\gto_\gitt$. We denote by $\gcap_\gitt$ the meet, by $\equiv_\gitt$
the equivalence induced by  $\gleq_\gitt$, and we use $\agi$, $\bgi$ to range over the elements of $\gitt$. 
\end{definition}

Proposition~\ref{semilattice} shows that an itt yields naturally a gitt.  Notably there are gitt's which are not itt's, an example is the gitt $\Pair{\SL}\subseteq$ defined in Theorem~\ref{msts}. 

\medskip

We introduce the following notion of morphism between gitt's.

\begin{definition}[Embedding]\label{gembd} Let $\Pair\gitt {\gleq_\gitt}$ and $\Pair{\gitt'} {\gleq_{\gitt'}}$ be two gitt's, then $\Pair\gitt {\gleq_\gitt}$ is {\em embeddable} in $\Pair{\gitt'} {\gleq_{\gitt'}}$ if there is a function $\gm:\gitt\rightarrow\gitt'$ such that:
\begin{enumerate}
\item\label{gemd0} $\gm(\alpha) = \top_{\gitt'}$ if and only if $\alpha  \equiv_{\gitt} \top_{\gitt}$;
\item \label{gemd0,1} $\gm(\alpha \gto_\gitt \beta ) = \gm(\alpha) \gto_{\gitt'} \gm(\beta)$;
\item \label{gemd0,2}$\gm(\alpha \gcap_\gitt \beta ) = \gm(\alpha) \gcap_{\gitt'} \gm(\beta)$;
\item\label{gembd3} $\agi\gleq_\gitt \bgi$ implies $\gm(\agi)\gleq_{\gitt'} \gm(\bgi)$.
\end{enumerate}
\end{definition}

We can naturally extend the notion of Type Assignment Systems to gitt's  following~\cite{DDH26}. 

\begin{definition}[Generalised Type Assignment System]\label{gtas}
The {\em intersection type assignment system} induced by a  gitt $\Pair\gitt {\gleq_\gitt}$ is a formal system deriving judgements of the shape
$\gder\gb M\agi$, where $\agi\in\gitt$ and a {\em basis} $\gb$ is a finite mapping from term variables to elements in $\gitt$:
\Cline{\gb::=\emptyset\mid\gb, x:\agi.} 
The axioms and rules of the type system are the following
\Cline{\begin{array}{ccc}
\prooftree
\justifies
\gb,x:\agi\vdash x:\agi\using\rn{(Ax)}
\endprooftree
&\qquad&
\prooftree
\justifies\gb\vdash M:\top_\gitt\using\rn{($\top$)}
\endprooftree
\\[5mm]
\prooftree
\gb, x:\bgi\vdash M:\agi
\justifies 
\gb\vdash\lambda  x. M:\bgi\gto_\gitt\agi
\using\rn{($\gto$I)}
\endprooftree
&&
\prooftree
\gb\vdash M:\bgi\gto_\gitt\agi\quad\gb\vdash N:\bgi
\justifies
\gb\vdash M N:\agi
\using\rn{($\gto$E)}
\endprooftree\\[7mm]
\prooftree
\gb\vdash M:\bgi\quad\gb\vdash M:\agi
\justifies
\gb\vdash M:\bgi\gcap_\gitt\agi
\using\rn{($\gcap$I)}
\endprooftree
&&
\prooftree
\gb\vdash M:\bgi\quad \bgi\gleq_\gitt\agi
\justifies
\gb\vdash M:\agi
\using\rn{($\gleq$)}
\endprooftree
\end{array}}
\end{definition}

It is now natural to extend to gitt's also the notions of filter and filter model, and then it is straightforward to extend all results on itt's in Section~\ref{itfm} also to gitt's.

\medskip

A gitt $\Pair\gitt {\gleq_\gitt}$ is {\em sensible} if any unsolvable term has only types equivalent to $\top_\gitt$. 

\bigskip

The following transfer theorem, will prove very useful in the sequel:

\begin{theorem}[Transfer~{\cite[Theorem 8]{DDH26}}]\label{gembp} Let $\Pair\gitt {\gleq_\gitt}$ be {\em embeddable} in $\Pair{\gitt'} {\gleq_{\gitt'}}$. We get:
\begin{enumerate}
\item\label{gembp1}
if $\Pair{\gitt'} {\gleq_{\gitt'}}$ is sensible, then $\Pair\gitt {\gleq_\gitt}$ is sensible. 
\item\label{gembp2}
if  $\Pair\gitt {\gleq_\gitt}$ is non-sensible, then $\Pair{\gitt'} {\gleq_{\gitt'}}$ is non-sensible.
\end{enumerate}
\end{theorem}

We will now construe Tait-Girard reducibility candidates as a gitt.
Let $\So$ denote the set of solvable terms.

\begin{definition}
A set $X\subseteq\So$ is {\em saturated} if it is closed under  $\beta$-conversion  and $x\overrightarrow M\in X$ for all $x$, $\overrightarrow M$.
\end{definition}

\noindent
The set of saturated sets is a complete lattice w.r.t. set inclusion, with  $\Bo=\set{M\in \Lambda\mid M\bredstar x\overrightarrow M}$ as bottom  and $\So$ as top.  Notice that $\Bo$ is closed under $\beta$-conversion, since $\beta$-reduction enjoys the Church-Rosser property.  We use $\SAT$ to denote this lattice and  $\Opp\SAT$ to denote $\SAT$  with the reverse order having bottom $\So$ and top $\Bo$. 
Formally

\Cline{\SAT=\Pair{\set{\Bo\subseteq X\subseteq\So\mid X\text{ is saturated}}}{\subseteq}\text{ and  }\Opp\SAT=\Pair{\set{\Bo\subseteq X\subseteq\So\mid X\text{ is saturated}}}{\supseteq}.} 

We define
\Cline{X\Rightarrow Y=\set{M\in\Lambda\mid \forall N\in X~MN\in Y},}
where $X$ and $Y$ range over saturated sets or $\Lambda$. 
It is easy to verify that $Y$ saturated implies $X\Rightarrow Y$ saturated and that $X\Rightarrow \Lambda=\Lambda$. Moreover $X$ and $Y$ saturated imply $X\cap Y$ saturated and $X\cap\Lambda=X$ for all  $X$. 

We define $\SL=\set{\Bo\subseteq X\subseteq\So\mid X\text{ is saturated}}\cup\set{\Lambda}$.

The following result is crucial.

\begin{theorem}\label{msts}
The gitt %$\SL$, namely 
$\Pair{\SL}\subseteq$ with top $\Lambda$, meet $\cap$ and arrow $\Rightarrow$ is sensible.
\end{theorem}
\begin{proofs}  We show that  $\sder{\gb}{M}{X}$ with $X\subsetneq\Lambda$ implies that $M\in X$ hence solvable.  Let $\gb=\set{x_i:Y_i\mid 1\leq i\leq n}$  and $N_i \in Y_i$ for $1\leq i\leq n$. By induction on the type derivations we can  prove that 
$\sder{\gb}{M}{X}$ with $X\subsetneq\Lambda$ implies \Cline{M[x_i:=N_i\mid 1\leq i\leq n]\in X.} 
The most interesting case is when the last applied rule is Rule ($\gto$I). In this case $M=\lambda x.M'$ and $X=Y\Rightarrow X'$ and $\sder{\gb,x:Y}{M'}{X'}$. Let $N\in Y$, then by induction hypothesis we have that \Cline{M'[x_i:=N_i\mid 1\leq i\leq n][x:=N]\in X',} which implies $M[x_i:=N_i\mid 1\leq i\leq n]N\in X'$, since saturated sets are closed under $\beta$-conversion. Since $N\in Y$ is arbitrary by definition we conclude \Cline{M[x_i:=N_i\mid 1\leq i\leq n]\in Y\Rightarrow X'=X.\hspace{50pt}\square}
\end{proofs}

\bigskip

The rest of this section is devoted to giving examples of how to apply Theorem~\ref{gembp}, above, to some  itt's and gitt's thereof. It gives a very flexible criterion which cuts both ways, since it can be used both  for {\em reducing} the sensibility of a filter model to that of the embedded filter model,  but also for {\em extending}   the non-sensibility of the embedded filter model to that of the filter model in which it embeds.  
\begin{example}\label{emb}
\hfill
\begin{enumerate}
\item\label{emb2} The itt $\ITT_{BCD}$,  defined in~\cite{BCD83},  is shown  to be sensible by normalising type derivations. The   not  $\beta$-sound itt's defined in~\cite{CDHL84,ABD03} are sensible, since they can be embedded in the itt $\ITT_{BCD}$ by equating the two distinguished constants in such a way that the resulting subtyping amounts precisely to $\leq_{BCD}$. 
\item \label{emb3}
The $\to$-sound itt $\ITT_*$ with the constant $\vartype$ and the axiom $\vartype\leq\vartype\to\vartype$ is defined in~\cite{A08}. Similarly we can define the $\to$-sound itt $\ITT^*$ with the constant $\vartype$ and the axiom  $\vartype\to\vartype\leq\vartype$. 
These itt's are sensible, since they can be appropriately embedded in the sensible itt $\ITT_{CDZ}$ defined in~\cite{CDZ87}. The  
itt $\ITT_{CDZ}$ has only two totally ordered constants. The embedding is realised by interpreting $\vartype$ as the smaller
constant in the former case and as the bigger constant in the latter case. 
\item \label{emb4}
The $\to$-sound itt $\ITT^\flat$ with the constant $\vartype$ and the axiom $\vartype\sim(\vartype\to\vartype)\cap\vartype$ can be again embedded in the sensible itt $\ITT_{CDZ}$ defined in~\ref{emb3}. This embedding is realised by mapping $\vartype$ to the smaller constant.
\end{enumerate}
\end{example}

\noindent
The next proposition makes it possible to build non-sensible filter models starting from sensible ones.

\begin{proposition}
If $\ITT$ is a sensible itt we can define a non-sensible itt $\ITT'$ such that $\ITT$ is embeddable in $\ITT'$.
\end{proposition}
\begin{proof}
Let $\ITT=\Pair{\cSet_\ITT}{\leq_\ITT}$. We define $\ITT'$ by adding a constant $\vartype\not\in\cSet_\ITT$ and the axiom $\vartype\sim\vartype\to\vartype$. The embedding is the identity. In this way $\ITT'$  is obtained from $\ITT$ essentially by adding the itt generating the filter model isomorphic to Park model~\cite{P76} defined  in~\cite{HR92}.
\end{proof}

This proposition permits us to build neither sensible nor $\beta$-sound filter models  starting from the filter models defined in~\cite{CDHL84,ABD03}.

\section{Morphisms Engineering}\label{gt}
% !TEX root =dgh.tex
 
The power of the Transfer Theorem~\ref{gembp} in proving sensibility of itt's, or more generally gitt's, derives from the existence of appropriate embeddings in $\SAT$. Historically, this was done implicitly by defining appropriate 
%As expected the embedding of itt's in $\SL$ is a powerful tool to show the sensibility of itt's. In this case the embedding function can be given by a
 type interpretations based on Tait-Girard's computability arguments in~\cite{T67,G71,CDZ87,K90,HL99}. 

In this section, we discuss two conditions on itt's, or gitt's derived thereof, which ensure that appropriate morphisms, yielding sensibility, exist. The first condition, Definition~\ref{sitt}, is not effective and it is an almost trivial reformulation of the results in the previous section. Its interest lies in that it can be reversed, Theorem~\ref{sesa},  for a very large class of itt's, including inverse limit models, thus showing that $\SAT$ is somewhat {\em universal}. The second condition, Definition~\ref{consistent-polarity}, is a reformulation of Mendler's condition~\cite{M91} to intersection type theories, and allows for showing constructively the sensibility of many itt's. 

\medskip

Since most of the gitt's in this section  arise  from itt's, we shall reason directly on itt's.

\medskip            

An $\cSet$-environment is a mapping from  a set of type constants, $\cSet$,  into $\SL$. We use $\aenv$ to range over $\cSet$-environments. 

\begin{definition}[Type Interpretation]\label{ti}
The {\em type interpretation} of the set of intersection types $\IT_\cSet$ induced by the $\cSet$-environment $\aenv$, notation $\tint\type\aenv$, is defined by:
\Cline{
\tint\U\aenv=\Lambda \qquad
\tint\vartype\aenv=\aenv(\vartype)\qquad
\tint{\type\to\types}\aenv=\tint{\type}\aenv\Rightarrow\tint{\types}\aenv\qquad
\tint{\type\cap\types}\aenv=\tint{\type}\aenv\cap\tint{\types}\aenv
.
}
\end{definition}

\noindent
Notice that either $\tint\type\aenv$ is a saturated set or $\tint\type\aenv=\Lambda$.

\begin{definition}[Saturation]\label{sitt}
An itt $\ITT$ is {\em saturated} if there is a type interpretation  which gives rise to a morphism in the sense of Definition~\ref{gembd} between the gitt induced by $\ITT$ and $\Pair{\SL}\subseteq$.
\end{definition}

It is easy to verify that all conditions of Definition~\ref{gembd} are satisfied by type interpretations, but for condition~\ref{gemd0} which requires that
Axiom \rn{($\to\U$)} holds in $\ITT$. Then, Theorems~\ref{gembp}(\ref{gembp1}) and~\ref{msts} immediately imply that:

\begin{theorem}\label{mstsi}  A saturated itt is sensible. \end{theorem}

\begin{figure}[t!]

%\small
\Cline{\begin{array}{lll}
\tint{\types\to\type}{\widehat\aenv}&=&\tint{\types}{\widehat\aenv}\Rightarrow\tint{\type}{\widehat\aenv}	%\\
%&&&
\hfill \text{by Definition~\ref{ti}}\\
&=&\set{M\mid \forall N\in\tint{\types}{\widehat\aenv}~~MN \in\tint{\type}{\widehat\aenv}}%\\
%&&
\hfill \text{by Definition of $\Rightarrow$} \\
&=&\set{M\mid \forall N ~~\exists \Gamma~~\der \Gamma N \types~~\der \Gamma {MN} \type}%\\
%&&
\hfill  \text{by induction}  \\
&=&\set{M\mid \exists \Gamma~~\der {\Gamma,x:B} {Mx} \type}%\\
%&&
\hfill \text{where $x$ is fresh by Lemma~\ref{aux}}\\
&=&\set{M\mid \exists \Gamma~~\der {\Gamma} {M} {\types_i\to\type_i}~~\der{x:\types} x{\types_i}~~ \forall i\in\pI~~
\bigcap_{i\in\pI}\type_i\leq_\ITT\type}\\
&&\hfill \text{by Lemma~\ref{il}(\ref{il2})}\\
&=&\set{M\mid \exists \Gamma~~\der {\Gamma} {M} {\types_i\to\type_i}~~\types\leq_\ITT\types_i~~ \forall i\in\pI~~
\bigcap_{i\in\pI}\type_i\leq_\ITT\type}\\
&&\hfill \text{by Lemma~\ref{il}(\ref{il1})}\\
&=&\set{M\mid \exists \Gamma~~\der {\Gamma} {M} {\types\to\type_i}~~\forall i\in\pI~~
\bigcap_{i\in\pI}\type_i\leq_\ITT\type}\\
&&\hfill \text{by Rule \rn{($\leq$)} using Rule \rn{($\to$)}}\\
&=&\set{M\mid \exists \Gamma~~\der {\Gamma} {M} {\bigcap_{i\in\pI}(\types\to\type_i)}~~
\bigcap_{i\in\pI}\type_i\leq_\ITT\type}%\\
%&&
\hfill \qquad\qquad \text{ by Rule \rn{($\cap$I)}}\\
&=&\set{M\mid \exists \Gamma~~\der {\Gamma} {M} {\types\to\bigcap_{i\in\pI}\type_i}~~
\bigcap_{i\in\pI}\type_i\leq_\ITT\type}\\
&&\hfill \text{by Rule \rn{($\leq$)} using Axiom \rn{($\to\cap$)}}\\
&=&\set{M\mid \exists \Gamma~~\der {\Gamma} {M} {\types\to\type}}%\\
%&&
\hfill \text{by Rule \rn{($\leq$)} using Rule \rn{($\to$)}.}\end{array}}

\medskip

\Cline{
\begin{array}{lll}
\tint{\types\to\U}{\widehat\aenv}&=&\tint{\types}{\widehat\aenv}\Rightarrow\tint{\U}{\widehat\aenv}%\\
%&&
\hfill  \text{by Definition~\ref{ti}} \\
&=&\set{M\mid \forall N\in\tint{\types}{\widehat\aenv}~~MN \in\tint{\U}{\widehat\aenv}}%\\
%&&
\hfill  \qquad\qquad\text{ by Definition of $\Rightarrow$}\\
&=&\set{M\mid \forall N ~~ \exists\Gamma~~\der \Gamma N \types~~\der \Gamma {MN} \U}%\\
%&&
\hfill  \text{by induction} \\
&=&\Lambda%\\
%&&
\hfill \text{by Rule \rn{$(\U)$}}\\
&=&\set{M\mid \ \der {} {M} \U}%\\
%&&
\hfill \text{by Rule \rn{$(\U)$}}\\
&=&\set{M\mid \ \der {} {M} \types\to\U}%\\
%&&
\hfill \text{by Rule \rn{($\leq$)} using Axiom \rn{($\to\U$)}}.
\end{array}}
\caption{Proof of Theorem~\ref{sesa}.}\label{prova}
\end{figure}

For $\to$-sound itt's, Theorem~\ref{mstsi} can be reversed. For this we  first need an easy auxiliary lemma.

\begin{lemma}\label{aux} If $\der{\Gamma,x:\types}{Mx}\type$ where $x$ does not occur in $M$ and $\der{\Gamma'}N\types$, then 
$\der{\Gamma''}{MN}\type$ for some $\Gamma''$.
\end{lemma}
\begin{proof} Define 
\Cline{\Gamma_1\Cup\Gamma_2=\set{y:\typeC_1\cap\typeC_2\mid y:\typeC_1\in \Gamma_1~~ y:\typeC_2\in \Gamma_2}\cup\set{y:\typeC_1\mid y:\typeC_1\in \Gamma_1 ~~ y\not\in\Gamma_2}\cup\set{y:\typeC_2\mid y:\typeC_2\in \Gamma_2 ~~ y\not\in\Gamma_1}. }
We can build a derivation of $\der{\Gamma''}{MN}\type$ just by replacing the axioms $\der{\hat\Gamma,x:\types}{x}\types$ with derivations of $\der{\hat\Gamma\Cup\Gamma'}N\types$ in a derivation of $\der{\Gamma,x:\types}{Mx}\type$. 
\end{proof}

\begin{theorem}\label{sesa}
Each $\to$-sound and sensible itt is saturated. 
\end{theorem}
 \begin{proof} Let  $\ITT=\Pair{\cSet}{\leq_\ITT}$ be an $\to$-sound and sensible itt.  Define the type interpretation \Cline{\widehat\aenv(\vartype)=\set{M\mid \exists \Gamma~~\der \Gamma M \vartype}.}
It is enough to show now that  $\tint\type{\widehat\aenv}=\set{M\mid \exists \Gamma~~\der \Gamma M \type},$  since all conditions of Definition~\ref{gembd} hold, and in particular  $\type\leq_\ITT\types$ implies \Cline{\set{M\mid \exists \Gamma~~\der \Gamma M \type}\subseteq\set{M\mid \exists \Gamma~~\der \Gamma M \types}.} The proof is
 by induction on the definition of type interpretation. The only two interesting cases are proved in Figure~\ref{prova}, where we assume  $\type\nsim_\ITT\U$. 
\end{proof}

Notice that Theorem~\ref{sesa} does not provide an effective characterisation of sensibility for $\to$-sound itt's, since the definition of $\cSet$-environment is not constructive {\em per se}. 

\medskip

The rest of  this section is devoted to showing that a special class of itt's, satisfying the {\em positive polarity} condition in Definition~\ref{consistent-polarity}, is sensible. 
In any case, this class of itt's, which we call {\em natural}, includes essentially all sensible itt's ever used explicitly in the literature. 

\begin{definition}[Natural Itt's]
An $\ITT=\Pair{\set{\vartype_i}_{i\in \pI}}{\leq_\ITT}$ is {\em natural} if  $\leq_\ITT$ satisfies  axiom \rn{($\to\U$)} in Figure~\ref{fig:axRules}, is determined by a set of axioms of a very special form, namely $\Eqt=\{\vartype_i\sim\type_i\}_{i\in\pI}$, and possibly some other  axioms and rules in Figure~\ref{fig:axRules}. Moreover we assume that each type constant occurs exactly once on the left hand side of an axiom in $\Eqt$, possibly vacuously as an identity.  The set $\Eqt$ is the {\em characteristic set} of $\ITT$. 
\end{definition}

%We can also assume that $\Eqt$ is closed, since we  can add the axion $\vartype\sim\vartype$ for the $\vartype$'s that that are not defined in $\Eqt$.

We are now in the position of giving the following crucial definition.

\begin{definition}[Positive Polarity Condition]\label{consistent-polarity}
A natural $\ITT=\Pair{\set{\vartype_i}_{i\in \pI}}{\leq_\ITT}$ satisfies the {\em positive polarity} condition if for all equations of the form $\vartype\sim \type$ derivable from the axioms/rules defining  ${\leq_\ITT}$, from $\OkPos(\type)$ we cannot derive $\OkNeg(\vartype)$ by
applying the following rules: \Cline{
\begin{array}{cccccccc} 
\prooftree
\OkPos(\type\to \types)
\justifies 
\OkNeg(\type)   
\endprooftree 
&\quad& 
\prooftree
\OkPos(\type\to \types)
\justifies 
\OkPos(\types)   
\endprooftree 
&\quad&   
\prooftree
\OkNeg(\type\to \types) 
\justifies 
\OkPos(\type)
\endprooftree
&\quad&   
\prooftree
\OkNeg(\type\to \types) 
\justifies 
\OkNeg(\types)
\endprooftree
\\ [20pt]
\prooftree
\OkPos(\type\cap \types)
\justifies 
\OkPos(\type)
\endprooftree  
&\quad&
\prooftree
\OkPos(\type\cap \types)
\justifies 
\OkPos(\types)
\endprooftree  
&\quad&
\prooftree
\OkNeg(\type\cap \types)
\justifies 
\OkNeg(\type)
\endprooftree 
&\quad&
\prooftree
\OkNeg(\type\cap \types)
\justifies 
\OkNeg(\types)
\endprooftree 
\end{array}
}
\end{definition}

\noindent
It is easy to check that this condition essentially amounts to the fact that if $\vartype\sim_{\ITT}\type$, the constant $\vartype$ does not occur in $\type$ nested inside an odd number of arrows.

\begin{example}
Let $\ITT^\#=\Pair{\set{\vartype_0,\vartype_1,\vartype_2}}{\leq_{\ITT^\#}}$, where $\leq_{\ITT^\#}$ has the only axioms $\vartype_0\sim\vartype_1\to\vartype_2$ and $\vartype_1\sim\vartype_0$. By Rule \rn{($\to\sim$)}  we derive  $\vartype_0\sim\vartype_0\to\vartype_2$ and $\OkPos(\vartype_0\to\vartype_2)$ implies $\OkNeg(\vartype_0)$. Therefore $\ITT^\#$ does not satisfy the positive polarity condition. 
\end{example}

From now on until the end of the section, unless otherwise stated, we will assume that itt's are natural and satisfy the positive polarity condition in Definition~\ref{consistent-polarity}.  Moreover for simplicity,  we consider only characteristic sets 
in which the axioms are of the following three forms: $\vartype\sim\vartype$, or  $\vartype\sim\vartype'\to\vartype''$, or $\vartype\sim\vartype'\cap\vartype''$. In fact we can always transform sets of axioms in this form by removing renamings
and adding new constants and axioms to simplify the right-hand-side of the original axioms. 

To prove that  $\ITT=\Pair{\set{\vartype_i}_{i\in \pI}}{\leq_\ITT}$ is saturated  we have to find an  $\cSet$-environment $\aenv$ which 
%maps associating elements of $\SAT$ to the constants in $\set{\vartype_i}_{i\in\pI}$ 
induces a morphism. Since $\Rightarrow$ on saturated sets is contra-variant on the domain and covariant on the co-domain, the set condition is harmless since it only allows for less set inclusions then those which hold in every type interpretation. 
Moreover we have the following proposition. 
 
\begin{proposition}
Every type interpretation satisfies the axioms and rules in Figure~\ref{fig:axRules}.
\end{proposition}
\begin{proof}
We only consider two interesting cases. Rule \rn{($\to$)} follows from the contra-variance/covariance of $\Rightarrow$. Rule \rn{($\U\leq$)} follows from the fact that
$X\Rightarrow Y=\Lambda$ implies $Y=\Lambda$.
\end{proof}

The natural idea to find a type interpretation for a natural itt in $\SAT$, would be to define, out of the  characteristic set, a monotone operator and use the fact that $\SAT$ is a complete lattice and hence by Knaster-Tarski's Theorem each monotone operator has a complete lattice of fixed points. But the positive polarity condition, Definition~\ref{consistent-polarity}, yields only an individual constraint on each type constant, which cannot be extended uniformly. Conflicting polarities would naturally arise as in the case of the itt's in the following example.  
\begin{example}\label{ep}
Let $\Eqt=\Eqt'\cup\Eqt''$, where $\Eqt'=\set{\vartype_1\sim \vartype_2\to \vartype_1,\vartype_2\sim \vartype_1\to \vartype_2}$ and $\Eqt''=\set{\vartype_3\sim\vartype_4\cap\vartype_5, \vartype_4\sim\vartype_1\to\vartype_3, \vartype_5\sim\vartype_2\to \vartype_3}$. The axioms in $\Eqt'$ require that $\vartype_1$ and $\vartype_2$ have opposite polarities, while the axioms in $\Eqt''$ require that $\vartype_1$ and $\vartype_2$ have the same polarity. This example will be further discussed in Example~\ref{ex:closure}.
\end{example}

In order to be able to define an appropriate type interpretation we need therefore to  introduce an appropriate order on the constants appearing in the characteristic set, so that they can be progressively dealt with. To this end we need a number of definitions.

\begin{definition}[Completion, Closure, Equivalence Class] \label{equivalenza} Consider an itt $\ITT=\Pair{\set{\vartype_j}_{j\in J}}{\leq_\ITT}$ and a subset $\Eqt=\{\vartype_i\sim\type_i\}_{i\in\pI}$ of its characteristic set.
\begin{enumerate}
\item We say that axiom $\vartype\sim\type$ {\em defines} the constant $\vartype$. Hence the
%\item The 
{\em set of constants} defined in $\Eqt$,  notation $\CA\Eqt$, is $\{\vartype_i\}_{i\in\pI}$.
%\item We define the {\em projection} of the set of axiom $\Eqt$ on the set of constants $\cSet\subseteq\CA\Eqt$, notation $\pr\Eqt\cSet$, by $\set{\vartype\sim\type\in\Eqt\mid\vartype\in\cSet}$.
%\item $\Eqt$ is {\em full} if $\CA\Eqt$ contains all constants which are hereditarily defined using a variable in $\CA\Eqt$.
\item\label{two}
The {\em completion} of a full set of axioms $\Eqt$ is $\Eqt\cup\set{\vartype\sim\vartype\mid \vartype'\sim\type\in \Eqt~\&~ \vartype\text{ occurs in }\type ~\&~ \vartype\not\in\CA\Eqt }$. 
%\item 
A set of axioms which coincides with its completion is {\em complete}.
%\item The {\em projection} of the set of axiom $\Eqt$ on the set of constants $\cSet\subseteq\CA\Eqt$, notation $\pr\Eqt\cSet$, is the set $\set{\vartype\sim\type\in\Eqt\mid\vartype\in\cSet}$.
%\item  Given a set  of constants $\cSet\subseteq\CA\Eqt$, $\cSet$ is {\em closed} for $\Eqt$ if the only constants in $\CA\Eqt$ involved in the axioms defining constants in $\cSet$ belong to $\cSet$.\item  Let $\vartype\in\CA\Eqt$.  
\item Let $\Eqt$ be complete and $c\in\CA\Eqt$. \begin{enumerate}\item  The {\em closure} of $\vartype$ for $\Eqt$, notation $\cl\vartype\Eqt$, is ${\mathcal C}(\Eqt')$, where $\Eqt'$ is the smallest complete subset of $\Eqt$ such that $\vartype\in\CA{\Eqt'}$.
%which contains $\vartype$ and is closed for $\Eqt$.
\item The {\em equivalence class} of $\vartype$ for $\Eqt$, notation $\cls\vartype\Eqt$, is defined by \Cline{\set{\vartype'\in\CA\Eqt\mid \cl\vartype\Eqt=\cl{\vartype'}\Eqt}.}
%\item The {\em elimination} of  an equivalence class from a set of axioms is given by the set $\Eqt\setminus\cls\vartype\Eqt=\set{\vartype'\sim\type\in\Eqt\mid\vartype'\not\in\cls\vartype\Eqt}$.
\end{enumerate}
\end{enumerate}
\end{definition}
\noindent
The reason for the seemingly tautological clause~\ref{two} is to turn $\Eqt$ into a characteristic set.

\begin{example}
Let $\ITT_{\#}=\Pair{\set{\vartype_0,\vartype_1,\vartype_2}}{\leq_{\ITT_{\#}}}$ and $\Eqt_{\#}=\set{\vartype_0\sim\vartype_2\to\vartype_0,\vartype_1\sim\vartype_2\to\vartype_1}$. Then $\CA{\Eqt_{\#}}=\set{\vartype_0,\vartype_1}$ and the completion of 
$\Eqt_{\#}$ is $\set{\vartype_0\sim\vartype_2\to\vartype_0,\vartype_1\sim\vartype_2\to\vartype_1,\vartype_2\sim\vartype_2}$. \end{example}

\noindent
Clearly equivalence classes for a complete $\Eqt$ induce an equivalence relation on constants parameterised on $\Eqt$, namely, $\vartype\equiv_\Eqt\vartype'$ if 
$\cls\vartype\Eqt=\cls{\vartype'}\Eqt$. We can thus define the following relation, which is a well-defined partial order.

\begin{definition}[Partial Order]\label{po} Let $\Eqt$ be complete and $\vartype, \vartype'\in\CA\Eqt$.  The {\em partial order} between equivalence classes for $\Eqt$ is defined by $\cls\vartype\Eqt\leqc\cls{\vartype'}\Eqt$ if $\cl{\vartype'}\Eqt\cap\cls\vartype\Eqt\neq
\emptyset$.
\end{definition}

\begin{example}\label{ex:closure}
Let $\Eqt$, $\Eqt'$ and $\Eqt''$ be as in Example~\ref{ep}. 
%The set of constants $\cSet'=\set{\vartype_1,\vartype_2}$ is closed for $\Eqt$ and so is
%$\cSet=\set{\vartype_1,\vartype_2,\vartype_3,\vartype_4,\vartype_5}$.
Both $\Eqt$ and $\Eqt'$ are complete, while $\Eqt''$ is not. 
 Moreover $\cl{\vartype_1}\Eqt=\cl{\vartype_2}\Eqt=\CA{\Eqt'}$ and
$\cl{\vartype_3}\Eqt=\cl{\vartype_4}\Eqt=\cl{\vartype_5}\Eqt=\CA{\Eqt}$. So the axioms in $\Eqt$ define two equivalence classes: $\cls{\vartype_1}\Eqt=\CA{\Eqt'}$
%\set{\vartype_1,\vartype_2}$
and $\cls{\vartype_3}\Eqt=\set{\vartype_3,\vartype_4,\vartype_5}$, ordered by $\cls{\vartype_1}\Eqt\leqc\cls{\vartype_3}\Eqt$.
\end{example}

\medskip

We are now in the position of proving the main result, Theorem~\ref{mt}, namely that a natural itt whose characteristic set of axioms satisfies the positive polarity condition, in Definition~\ref{consistent-polarity}, cannot type an unsolvable term. We do this in three steps.
\begin{enumerate}
\item We restrict to natural type theories which are {\em finite}. That this kind of {\em compactness} result is enough for dealing even with infinite sets of axioms was first noticed by Mendler~\cite{M91}, since all but a finite number of constants are ever used in any type derivation. %In fact, 
Moreover, if the defining equation of a constant is not used in a derivation where that constant appears, then that constant can be safely taken to be equal just to itself. 
\item We show how to give a type interpretation for a complete set of axioms $\mathcal A$ such that $\CA\Eqt$ consists of a single equivalence class for $\Eqt$, Proposition~\ref{minimal}. 
\item We show how to extend a given type interpretation for a complete set of axioms $\Eqt$ to a  type interpretation for the larger complete set of axioms $\Eqt'$ such that the added constants have all the identity axiom in $\Eqt'$,
Proposition~\ref{inductive step}.
\end{enumerate}

\noindent
Both Propositions~\ref{minimal} and~\ref{inductive step} are proved  exploiting the fact that  complete subsets of $\Eqt$ define appropriate
monotone operators on the complete lattice $\Pi_{i\in\pI}\Sat_i$, where $\Sat_i$ can be either $\SAT$ or $\Opp\SAT$.  Then any fixed point  of these
operators, which we know to exist,  provides the tuple of saturated sets giving rise to the $\zeta_\cSet$-environment which we need. 

In order to define the operators we first need to decorate constants in the axioms $\Eqt=\{\vartype_i\sim\type_i\}_{i\in\pI}$ 
 with a polarity $\pol\in\set{\piu,\meno,\pom}$. 
 %We denote by $\overline\pol$ the opposite polarity to $\pol$, where $\overline\pm=\pm$. 
 The intuition is that the axiom associated to a constant
$\Pos\vartype$ should define an operator that is monotone in  $\SAT$ on the variable corresponding to that constant, and the one
associated to a constant $\Neg\vartype$ should define an operator that is monotone in  $\Opp\SAT$  on the variable corresponding to that constant.
The decoration $\pom$ is used for constants whose axiom is the identity. 
%defines an operator that is monotone for  both $\SAT$ and  $\Opp\SAT$ on the corresponding variable. 
The polarity of constants can be extended in a natural way to all types built using them.

\begin{definition}[Polarity]\label{pol}
The predicates $\OkPos$ and $\OkNeg$ on types with polarised constants are defined by:
\Cline{
\begin{array}{c} 
\OkPos(\Pos\vartype)\quad\OkPos(\PN\vartype)\quad\quad \quad\OkNeg(\Neg\vartype)\quad\OkNeg(\PN\vartype)\\[10pt]
\prooftree
\OkNeg(\type)   \ \ \ \OkPos(\types)
\justifies 
\OkPos(\type\to \types)
\endprooftree 
\quad   
\prooftree
\OkPos(\type)   \ \ \ \OkNeg(\types)  
\justifies 
\OkNeg(\type\to \types)
\endprooftree
\quad \quad \quad 
\prooftree
\OkPos(\type)   \ \ \ \OkPos(\types)
\justifies 
\OkPos(\type\cap \types)
\endprooftree  
\quad
\prooftree
\OkNeg(\type)    \ \ \ \OkNeg(\types) 
\justifies 
\OkNeg(\type\cap \types)
\endprooftree 
\end{array}
}
\end{definition}

\begin{definition}
A decoration of constants, $\set{\Pol{\vartype}{\pol_i}_i}_{i\in\pI}$, {\em agrees} with a set of axioms $\Eqt=\{\vartype_i\sim\type_i\}_{i\in\pI}$ if $\Pos\vartype_i\sim\type_i$ implies $\OkPos(\type_i)$,
$\Neg\vartype_i\sim\type_i$ implies $\OkNeg(\type_i)$ and $\PN\vartype_i\sim\type_i$ implies $\type_i=\vartype_i$.
\end{definition}

Let $\B=\{\tilde\vartype_i\sim\type_i\}_{i\in\pI}$ be a complete set of axioms whose type constants are all in the same equivalence class for $\B$, and let $\set{\Pol{\tilde\vartype}{\pol_i}_i}_{i\in\pI}$ be a decoration of the constants which agrees with  $\B$. Let 
$\leqS$ denote $\subseteq$ for $\SAT$ and $\supseteq$ for $\Opp\SAT$. It is easy to see that there exists a decoration, by the positive polarity condition in Definition~\ref{consistent-polarity}, where moreover no constant is decorated with $\pm$.
Consider the lattice $(\Pi_{i\in\pI}\Sat_i,\leqS)$ where $\Sat_i=\SAT$ if $\pol_i=\piu$  and $\Sat_i=\Opp\SAT$ if $\pol_i=\meno$
and $\leqS$ is the order induced on the cartesian product by the order on its components.
That is  $\tseq{X_i}{i \in I}\leqS\tseq{X'_i}{i \in \pI}$, if, for all $i\in\pI$, $X_i$ and $X'_i$ 
are saturated sets and $X_i\leqS X'_i$. 
Let $\vs X$ range over variables. Define
the operator associated to $\B$,  $\op{\B}:\Pi_{i\in\pI}\Sat_i\to\Pi_{i\in\pI}\Sat_i$, by
\Cline{
\op{\B}(\tseq{\vs X_i}{i \in \pI})=\tseq{{\transl{\type_i}}} {i \in \pI}   
}
where the mapping $~\transl{\_}$ is defined by:
\Cline{
\transl{\type}=
\begin{cases}
\vs X_i& \text{if }\type={\tilde\vartype}^\pol_i\\
\vs X_j\Rightarrow\vs X_k & \text{if }\type={\tilde\vartype}^\pol_j\to {\tilde\vartype}^{\pol'}_k\\
\vs X_j\cap\vs X_k & \text{if }\type={\tilde\vartype}^\pol_j\cap {\tilde\vartype}^{\pol'}_k.
\end{cases}
}

\noindent
Then we can easily prove

\begin{proposition}\label{minimal} Let $\B=\{\tilde\vartype_i\sim\type_i\}_{i\in\pI}$ be a complete set of axioms whose type constants are all in the same equivalence class for $\B$, then the operator $\op{\B}$ defined above is monotone.
\end{proposition}

Let  $\B=\{\tilde\vartype_i\sim\type_i\}_{i\in I}\subseteq\Eqt$ and let $\CA\B$ be an equivalence class for $\Eqt$ such that the constants in $\B$ are either defined in $\B$ (\ie\ they belong to $\{\tilde\vartype_i\}_{i\in I}$) or they belong to $\{\tilde\vartype_j\}_{j\in J}$ with $I\cap J=\emptyset$ and we have already a type interpretation in $\SAT$ for them given by $\fp{X}_j$ with $j\in J$. Define ${\B^\Eqt}=\B\cup\{\tilde\vartype_j\sim\tilde\vartype_j\}_{j\in J}$. It is easy to see that, by the positive polarity condition in Definition~\ref{consistent-polarity},  there exists a decoration $\set{\Pol{\tilde\vartype}{\pol_h}_h}_{h\in\pI\cup\pJ}$  of the constants in $\CA{\B^\Eqt}$ which agrees with  ${\B^\Eqt}$, giving the polarity $\pm$ to all constants in $\{\tilde\vartype_j\}_{j\in J}$. %, and let $\leqS$ denote $\subseteq$ for $\SAT$ and $\supseteq$ for $\Opp\SAT$. 
Consider the lattice $(\Pi_{h\in\pI\cup\pJ}\Sat_h,\leqS)$ where $\Sat_h=\SAT$ if $\pol_h=\piu$ or $\pol_h=\pom$ and $\Sat_h=\Opp\SAT$ if $\pol_h=\meno$. Let
%and $\leqS$ is the order induced on the cartesian product by the order on its components.
%That is  
$\tseq{X_h}{h \in I\cup\pJ}\leqS\tseq{X'_h}{h \in \pI\cup\pJ}$ and $\vs X$ be as in previous case.
%, if, for all $i\in\pI$, $X_i$ and $X'_i$ 
%are saturated sets and $X_i\leqS X'_i$. 
%Let $\vs X$ range over variables. 
Define
the operator associated to ${\B^\Eqt}$,  $\op{{\B^\Eqt}}:\Pi_{h\in\pI\cup\pJ}\Sat_h\to\Pi_{h\in\pI\cup\pJ}\Sat_h$, by
\Cline{
\op{{\B^\Eqt}}(\tseq{\vs X_h}{h \in \pI\cup\pJ})=\tseq{{\transl{\type_h}}} {h \in \pI\cup\pJ}   
}
where the mapping $~\transl{\_}$ is defined by $\fp{X}_j$ if $\type={\tilde\vartype}^\pm_j$ and $\fp{X}_j$ is the solution for $\tilde\vartype_j$ with $j\in J$, and as in previous case otherwise. 
%\Cline{
%\transl{\type}=
%\begin{cases}
%\fp{X}_i& \text{if }\type={\tilde\vartype}^\pm_i \text{ and $\fp{X}_i$ is the solution for $\tilde\vartype_i$}\\
%\vs X_i & \text{if }\type={\tilde\vartype}^\pol_i\text{ with } \pol\neq\pm \\
%\vs X_j\Rightarrow\vs X_k & \text{if }\type={\tilde\vartype}^\pol_j\to {\tilde\vartype}^{\pol'}_k\\
%\vs X_j\cap\vs X_k & \text{if }\type={\tilde\vartype}^\pol_j\cap {\tilde\vartype}^{\pol'}_k
%\end{cases}
%}
We easily get

\begin{proposition}\label{inductive step} Let $\B=\{\tilde\vartype_i\sim\type_i\}_{I\in I}\subseteq\Eqt $ and let $\CA\B$ consist of an equivalence class for $\Eqt$ such that all the constants appearing in $\B$ either are in $\CA\B$ or are such that we already have a type interpretation for them. Then the operator $\op{{\B^\Eqt}}$ defined above is monotone.
\end{proposition}

We can now prove the main result,  which generalises Theorem 18 of~\cite{DDH26}. 
\begin{theorem}\label{mt} %A, possibly infinite,  natural itt 
A natural itt with a possibly infinite characteristic set satisfying the condition of positive polarity, in Definition~\ref{consistent-polarity}, is sensible.
\end{theorem}
\begin{proof} Consider a finite derivation in a natural  itt $\ITT$. Without loss of generality we can restrict to the finite natural itt $\ITT'$ whose characteristic set involves only the constants actually used in that derivation, possibly assigning the identity to constants whose defining axioms have not been used in the derivation. Now use Proposition~\ref{minimal} for one of the minimal equivalence classes, according to the partial order in Definition~\ref{po} on the constants in $\ITT'$, to derive a first partial type interpretation of the constants. Notice that the set of axioms defining the constants in a minimal equivalence class is complete. Use Proposition~\ref{inductive step} to extend such a type interpretation to all the constants in $\ITT'$  adding incrementally an equivalence class such that the solutions for the constants not belonging to that equivalence class have already be found. 
Since $\ITT'$ is finite, we can always find such an equivalence class, namely one of the minimal classes in the  partial order consisting of the equivalence classes which have not been yet dealt with. Finally, using Theorem~\ref{mstsi} we conclude the proof. 
\end{proof}

We end this section with a few examples. The sensibility of the first theory follows directly by applying Propositions~\ref{minimal} and~\ref{inductive step}. The second example deals with a type theory, which was introduced in~\cite{DGH25}. Its sensibility can be proved either using Theorem~\ref{mt} or even directly taking the fixed points of  a monotone operator defined on countable sequences of $\SAT$'s and  $\SAT^{op}$'s. Finally, the third example deals with a theory whose sensibility, to our present knowledge, can be proved only using Theorem~\ref{mt}, through its finite approximations. This is somewhat puzzling because once we know that the theory is sensible, by Theorem~\ref{sesa}, we can in principle define a type interpretation in $\SAT$. 

\begin{example}\ 
\begin{enumerate}
\item Consider the axioms $\Eqt$ of \refToEx{ep} and let $\cSet=\CA\Eqt$. 
\begin{itemize}
\item We start from $\cSet'=\cls{\vartype_1}\Eqt$, which is  the minimum class of $\Eqt$. %assigning polarity $\piu$ to $\vartype_1$. 
Let $\tseqs{ \fp X_1, \fp X_2} $ be a fixed point of the operator \Cline{\op{\cSet'}:\SAT\varotimes\Opp\SAT\to\SAT\varotimes\Opp\SAT} defined by 
 \Cline{\op{\cSet'}(\tseqs{\vs X_1,\vs X_2})=\tseqs{\vs X_2\Rightarrow \vs X_1,\vs X_1\Rightarrow \vs X_2}.} 
%The defining axioms for $\vartype_1$ is $\vartype_1\sim \vartype_2\to \vartype_1$, so we assign polarity $\meno$ to $\vartype_2$ and (re)assign $\piu$ to $\vartype_1$ . Now we have to analyse the defining axiom for $\vartype_2$ which is $\vartype_2\sim \vartype_1\to \vartype_2$ and (re)assign again polarity $\piu$ to $\vartype_1$ and $\meno$ to $\vartype_2$ and we are done.
\item We then analyse the class $\cls{\vartype_3}{{\Eqt}}$ taking advantage from the solutions for $\vartype_1$, $\vartype_2$ already computed.
We take as $\langle \fp X_1, \fp X_2,\fp X_3, \fp X_4, \fp X_5\rangle $ the fixed point of the operator \Cline{\op{\cSet}:\SAT\varotimes\SAT\varotimes\SAT\varotimes\SAT\varotimes\SAT\to\SAT\varotimes\SAT\varotimes\SAT\varotimes\SAT\varotimes\SAT} defined by 
 \Cline{\op{\cSet}(\tseqs{\vs X_1,\vs X_2,\vs X_3, \vs X_4, \vs X_5})=( \tseqs{\fp X_1, \fp X_2,\vs X_4\cap\vs X_5, \vs X_1\Rightarrow \vs X_3,\vs X_2\Rightarrow \vs X_3}).} 
 \end{itemize}
The sensibility of a natural itt with characteristic set $\Eqt$ can be shown by taking $\aenv(\vartype_i)=\fp X_i$ for  $1\leq i\leq 5$, where $\tseqs{\fp X_1, \fp X_2,\fp X_3, \fp X_4,\fp X_5}$ is a fixed point of $\op{\cSet}$.

\item Consider the axioms $\Eqt_\infty=\set{\vartype\sim\vartype}\cup\set{\vartype_n\sim\vartype_{n+1}\to\vartype\mid n \in \mathtt N}$ and $\cSet_\infty=\CA{\Eqt_\infty}$. The minimal class of $\Eqt_\infty$ is $\cls{\vartype}{{\Eqt_\infty}}$ and we can take for $\zeta_{{\Eqt_\infty}}(\vartype)$ an arbitrary saturated set, for example $\Bo$. But then there is no finite minimal equivalence class from which we can start our procedure for defining a type interpretation in $\SAT$. We could consider the finite approximations of a such a theory, but we can show also that a natural itt with characteristic set $\Eqt_\infty$ is sensible by defining directly the operator \Cline{\op{\cSet_\infty}:\SAT\varotimes(\SAT\varotimes\Opp\SAT)^{\mathbb N}\to\SAT\varotimes(\SAT\varotimes\Opp\SAT)^{\mathbb N}} by 
\Cline{\op{\cSet_\infty}(\tseqs{\vs X}\cdot\tseq{\vs X_n}{n\in{\mathbb N}})=\tseqs\Bo\cdot\tseq{\vs X_{n+1}\Rightarrow \vs X}{n\in{\mathbb N}}.}
 A fixed point of $\op{\cSet_\infty}$ exists, since it  is monotone. Let $\tseqs\Bo\cdot\langle {\fp X}_n \mid n\in{\mathbb N }\rangle $ be such a fixed point, then $\zeta_{{\cSet_\infty}}(\vartype)=\Bo$ and  $\zeta_{{\cSet_\infty}}(\vartype_n)=X_n$ for $n\in{\mathbb N}$ is the $\cSet_\infty$-environment we are looking for.

 \item Consider the itt given by the set of axioms $\set{\vartype_{0,n}\sim \vartype_{1,n} \to \vartype_{0,n+1},
\vartype_{1,n}\sim\vartype_{0,n} \to \vartype_{1,n+1}\mid n\in \mathtt N}$. This theory can be taken to be $\to$-sound and can be proved to be $\beta$-sound. Moreover finite approximations of this theory can be used to show its sensibility using Theorem~\ref{sesa}. We ignore how to define inductively an embedding of this theory in $\SAT$. 
\end{enumerate}
\end{example}

\section{Towards a Complete Characterisation of Sensible Itt's}\label{critical}
% !TEX root =dgh.tex

Mendler in~\cite{M91} studied second order $\lambda$-calculus with minimal and maximal fixed point type equations. He proved that the system is strongly normalising if and only if the fixed point equations satisfy essentially the positive polarity condition in  Definition~\ref{consistent-polarity}. 
Theorems~\ref{mstsi} and~\ref{sesa} are the analogues, albeit not effective, of Mendler's result, for $\rightarrow$-sound intersection type systems and solvable terms. The positive polarity condition on intersection type theories is only a sufficient condition for sensibility. We can indeed build a type interpretation which is, or finitely approximates, an embedding into $\SL$, for a natural  itt's whose characteristic set satisfies the positive polarity condition, but this is not a necessary condition as was the case in~\cite{M91}.  
The intersection operator $\cap$ can, in fact, sterilise the contra-variant behaviour of the  arrow constructor, as we can see in the following examples. All the itt's considered  in these examples are assumed to be $\to$-sound and moreover can be proved to be $\beta$-sound by induction on their subtypings. 
 
\begin{example}[Elimination of negative occurrences]\hfill\label{fine}
\begin{enumerate}
\item
Let $\ITT_2$  be the itt  with constants $\set{\vartype_0,\vartype_1}$  and  axiom
\Cline{
\vartype_0\sim\vartype_0\cap\vartype_1\to\vartype_{0}.} It is immediate to see that the characteristic set of $\ITT_2$ does not satisfy the positive polarity condition, in  Definition~\ref{consistent-polarity}. 

Nevertheless $\ITT_2$ can be shown to be sensible  by embedding it in the   itt $\ITT'_2$ obtained by adding  the axiom  $\vartype_1\leq\vartype_0$, which gives 
\Cline{
\vartype_0\sim_{\ITT_2'}\vartype_1\to\vartype_0
}
 generating  %resulting in 
a sensible filter model by  Theorem~\ref{mt}.  Alternatively, instead of adding the axiom $\vartype_1\leq\vartype_0$ we  can obtain a sensible filter model, again by   Theorem~\ref{mt},   by adding  the 
 axiom
\Cline{
\vartype_1\sim\U \rightarrow\vartype_{0}
}
since this axiom implies $\vartype_1\leq \vartype_0$  by Rule \rn{($\to$)}.  

\noindent
Notice, on the other hand, that if  we add to $\ITT_2$ the axiom $\vartype_0\leq\vartype_1$ we get $\vartype_0\sim_{\ITT_2}\vartype_0\to\vartype_0$. Then the resulting itt is non-sensible by Proposition~\ref{gembp}(\ref{gembp2}), because  the non-sensible itt generating the filter model isomorphic to Park model~\cite{P76} defined  in~\cite{HR92} is embeddable in it. 
\item 
Let $\ITT_3$  be 
the itt with constants $\set{\vartype_0,\vartype_1,\vartype_2}$ together with the axiom 
\Cline{
\vartype_0\sim\vartype_0\cap(\vartype_1\to\vartype_2)\to\vartype_1.
}
The itt 
$\ITT_{CDZ}$ considered in Example~\ref{emb}(\ref{emb3}) has constants $\set{\vartype_3,\vartype_4}$ and the axioms   
\Cline{
 \vartype_3\sim\vartype_4\to\vartype_3 \quad \vartype_4\sim\vartype_3\to\vartype_4\quad\vartype_3\leq\vartype_4.
}
We can show that $\ITT_3$ is sensible by embedding it in $\ITT_{CDZ}$ via the structural extension $\hat\iota$ of $\iota$ defined by:
\Cline{
\iota(\vartype_0)=\vartype_4\quad\iota(\vartype_1)=\vartype_4\quad\iota(\vartype_2)=\vartype_3.
}
 Now, since $\iota(\vartype_0)=\vartype_4$ and 
  \Cline{\hat\iota(\vartype_0\cap(\vartype_1\to\vartype_2))=\iota(\vartype_0)\cap(\iota(\vartype_1)\to\iota(\vartype_2))=\vartype_4\cap(\vartype_4\to\vartype_3)\sim_{\ITT_{CDZ}}\vartype_4\cap\vartype_3\sim_{\ITT_{CDZ}}\vartype_3} 
  which implies
  $\hat\iota(\vartype_0\cap(\vartype_1\to\vartype_2)\to\vartype_1)=\hat\iota(\vartype_0\cap(\vartype_1\to\vartype_2))\to\iota(\vartype_1)\sim_{\ITT_{CDZ}}\vartype_3\to\vartype_4\sim_{\ITT_{CDZ}}\vartype_4$,
  we have as required that  
\Cline{
\iota(\vartype_0)\sim_{\ITT_{CDZ}}\hat\iota(\vartype_0\cap(\vartype_1\to\vartype_2)\to\vartype_1).
}
\end{enumerate}
\end{example}

Achieving an {\em effective} Mendler-like completeness result appears critical even for natural intersection types and solvable terms, since there are cases where the intersection operator does not prevent the contra-variant behaviour of  the arrow constructor  to have the upper hand, as we can see in the following example.
\begin{example}\label{ff}
Let $\ITT_4$  be the  $\to$-sound itt %theory 
 with constants $\set{\vartype_0,\vartype_1,\vartype_2,\vartype_3}$
and %type system induced by the itt 
with the axiom 
\Cline{\vartype_0\sim\vartype_0\cap(\vartype_1\cap(\vartype_1\to\vartype_2)\to\vartype_2)\to\vartype_3.} 
We can type ${\boldsymbol\Omega}_2{\boldsymbol\Omega}_2$ with $\vartype_3$  since 
$\vdash_{\ITT_4}{\boldsymbol\Omega}_2:\vartype_0$ and $\vdash_{\ITT_4}{\boldsymbol\Omega}_2:\vartype_1\cap(\vartype_1\to\vartype_2)\to\vartype_2$.  The
$\beta$-soundness of $\ITT_4$ can be shown by induction on $\leq_{\ITT_4}$ and  hence  $\ITT_4$ generates a filter model. 
\end{example}
\bigskip

\subsection{Theories of Sensible Filter Models}\label{theories}

Models give semantics. But what are semantics? In the philosophical tradition crystallised by Leibniz, ontological entities arise once we can tell them apart. So semantics are essentially congruences. Given that there is a plethora of sensible filter models, we could imagine that these would provide  a corresponding plethora of semantics for $\lambda$-calculus, \ie\ $\lambda$-theories.  Formally a $\lambda$-theory is just a non-trivial congruence over $\lambda$-terms, closed under $\beta$-conversion. But this appears not to be immediately the case. 

All the $\lambda$-theories of sensible filter models which we have considered in this paper appear to equate all $\lambda$-terms which have the same B\"ohm tree, \ie\ their  $\lambda$-theories are at least $\mathcal B$. We refer to~\cite[Chapter 16]{B85} for more details on $\lambda$-theories and B\"ohm trees. This is the case for the filter model isomorphic to Scott's inverse limit model~\cite{CDHL84}, whose theory is the maximal sensible theory ${\mathcal H}^*$~\cite[Definition 16.2.1]{B85}, the filter model over $\ITT_{CDZ}$, defined in Example~\ref{emb}(\ref{emb3}), whose theory is the weaker ${\mathcal H}^+$~\cite[Definition 3.11(iii)]{BM22}, and of course the filter model over $\ITT_{BCD}$, defined in Example~\ref{emb}(\ref{emb2}), whose theory is ${\mathcal B}$~\cite[Definition 16.4.1]{B85}.
Notably in~\cite{DDH26} a sensible filter model which separates an open Barendregt  fixed-point combinator and the Curry fixed-point combinator is discussed. Therefore the theory of this filter model is weaker than $\mathcal B$.

The minimal sensible theory is $\mathcal H$~\cite[Definition 4.1.6(ii)]{B85}. It is an intriguing open problem whether this theory is precisely the theory of some filter model, or whether filter models have hitherto unknown semantical implications. We hope that this paper will stimulate readers to taking up this intriguing open question, which parallels for sensible theories the open question discussed in~\cite{HP09} for general  $\lambda$-theories. 

\bigskip 
 
We conclude this subsection discussing the notion of sensibility. \\
%% Frasi nuove
There are two notions under which a lattice model may be considered sensible. The more general notion requires that all closed unsolvable terms be identified~\cite[Definition 4.1.7(ii)]{B85}, whereas the more rigid notion requires that all unsolvable terms be identified with the bottom element of the model~\cite[Definition 12.1(ii)(3)]{BM22}.
In the present paper, we focused on the more specific notion of sensibility.
%% Fine modifica
%We ignore if they are equivalent. 
 We conjecture that by adding another universal constant, and a trivial rule, we can define  filter models where the unsolvables are not identified with the bottom element. The theory of such models, however, would still be order-sensible in the sense of~\cite[Definition 12.1(ii)(3)]{BM22}. A much more difficult problem would be to find a sensible filter model in the general sense which is non order-sensible. Given the non r.e. nature of unsolvables, this would quite likely  require non finitary  rules in the type assignment system.

 \section{Related Work and Conclusion}\label{rwc}
% !TEX root =dgh.tex

Since the invention  in the late seventies, 
%of intersection type theories in Torino, due chiefly to the first author of the present paper (or rather the discovery according to a different ontological standpoint), 
intersection types have revolutionised the approach to semantics of functional programming languages in multiple ways. Firstly, intersection types have reversed the traditional understanding of the relation of specifications to programs, justifying the correctness-oriented approach to program construction. Namely, we should use the specifications themselves to construct a program which meets them, rather than try to prove that  an existent program is  correct. This has been expressed categorically as a duality, see Abramsky~\cite{A91}, or by means of pointless topology~\cite{SVV96}. Secondly, intersection types have made explicit the connection between static and dynamic semantics, namely, the former semantics provides a finitary approximation of the latter. Thirdly, intersection types have allowed for static specifications of a plethora of interesting classes of $\lambda$-terms ~\cite{DHM05}. But, more generally, intersection types have provided, in the past half century, the paradigm for expressing and studying all sorts of semantics of programming languages ranging from quantitative semantics~\cite{DH93,G94,BKV17,AGK18,C18,AKR24}  to qualitative semantics~\cite{CDHL84,A91}, from games~\cite{HL03,DHL08,DL13} to power series~\cite{GO21}, and for all sorts of domains. 

Among the vast number of presentations  available today of intersection type theories, in this paper we have built upon the recent comprehensive discussion of filter models and unsolvable terms, which appears in~\cite{DGH25}. Actually, the present paper is a counterpart to that paper in that we discuss {\em sub specie typorum intersectionibus}, sensible filter models or, what is its syntactic analogue,  {\em head normalising} terms. 

Intersection type theories are very flexible and hence expressive, but this makes them also rather difficult to classify exhaustively. For instance the nice characterisation given by Mendler~\cite{M91}, of recursive second order type theories which type only  strongly normalising terms, cannot be paralleled in the context of  itt's and head normalising terms. There are plenty of itt's which do not satisfy any straightforward polarity criterion but nonetheless type non-trivially only head normalising terms. In~\cite{DDH26}  we argue that this is the case even for intersection-free axioms, contradicting blatantly the simple minded analogue of Mendler's condition.  {\em E.g.},  a natural 
 theory with the single axiom $\vartype\sim (((\vartype\rightarrow \vartype_0)\rightarrow \vartype_1)\rightarrow \vartype_2$ for $\vartype_0,\vartype_1,\vartype_2$ generic constants, can type  with types not equivalent to $\U$  only head-normalising terms.

In this paper, we construe itt's as special meet-semilattices and show that morphisms  in the  opposite category of meet-semilattices preserve sensibility, see Theorem~\ref{gembp}(\ref{gembp1}). Moreover we show that the meet-semilattice $\SAT$ is {\em universal} in the sense  that an $\to$-sound itt types non-trivially only head-normalising terms if and only if it can be embedded, as a meet-semilattice in it, see Theorems~\ref{mstsi} and~\ref{sesa}.  We provide a number of techniques for putting this result into action and give various examples. An immediate consequence is that sensibility transfers transitively in the op-category. Thus once we have a sensible itt, this  can play the role of  $\SAT$, and sensibility can be easily transferred to all itt's which embed in it.  Lacking suitable sensible itt's,  we need to define a direct morphism between an itt and $\SAT$. This can be achieved for a large class of natural itt's whose characteristic axioms satisfy a positive polarity condition. This condition essentially amounts to the condition introduced by Mendler in~\cite{M91} for second order $\lambda$-calculus. Thus, by repeatedly solving fixed point equations in $\SAT$, which is a complete lattice, we can prove Theorem~\ref{mt}, which amounts to the ``if'' part of Mendler's result. 

%Of course arguments involving indexed reduction can also be used~\cite{L75}. 

Providing a syntactical effective criterion for determining if an itt is sensible does not appear feasible, however, since intersections can produce rather unanticipated consequences, already in natural itt's. See the examples in Example~\ref{fine}. We have not studied itt's whose axioms are not equivalences or both whose sides are types.

In conclusion we have explored what was a ``seasoned" problem area and provided some advancement both in terms of conjectures and in terms of results. 

A last word goes to Stefano Berardi for whose 
birthday we dedicate the present paper.  We reckon Stefano Berardi not only among the most brilliant and deep Italian logicians of his generation, who has passed to-and-fro with breath-taking ability between Logic and Theoretical Computer Science, uncovering profound connections.  But we appreciate also his ethical attitude in current times.  Boldly, he has resisted  the fashion of pursuing quantity rather than quality, and of pursuing citations rather than results. Each of his works is original and extremely valuable in the mosaic produced by  the  noble and humble human endeavour of scientific research, of which he is a champion. Inspired by his attitude towards research we offer him the present paper as a small tile in that grand mosaic and a token of our friendship.

\paragraph{Acknowledgments}
The present version of this paper strongly improved with respect
to the original submission thanks to the careful reports.  The referees
did a great job in pointing out many places where the technical
details appeared without the needed explanations.  The difference
between the two versions are  several illustrating discussions with enlightening examples.

\bibliographystyle{plainurl}
 \bibliography{dgh}

\end{document}